\documentclass[9pt]{article}
\usepackage{comment}
\usepackage{amsthm,amstext,amssymb,amsmath,graphicx,amsfonts,caption,subcaption,multirow,epigraph,mathrsfs}
\usepackage[table,xcdraw]{xcolor}
\usepackage[all]{xy}
\usepackage[utf8]{inputenc}
\usepackage{booktabs}
\usepackage{tikz}
\usetikzlibrary{shapes.geometric, arrows.meta, positioning, calc, matrix, topaths, arrows}
\usepackage{placeins}
\usepackage{mathrsfs}
\usepackage{multicol}
\usepackage[numbers]{natbib}
\usepackage[pagebackref=false,colorlinks,linkcolor=blue,citecolor=blue]{hyperref}
\usepackage{tikz-cd}
\usepackage{longtable}
\usepackage{xcolor,soul}
\usepackage{float}
\usepackage[ruled,vlined]{algorithm2e}
\SetKwInOut{KwIn}{Input}
\SetKwInOut{KwOut}{Output}
\usepackage{pgfplots}
\pgfplotsset{compat=1.18}

\theoremstyle{plain}
\newtheorem{theorem}{Theorem}[section]

\theoremstyle{definition}

\newtheorem{example}[theorem]{Example}

\theoremstyle{remark}
\newtheorem{remark}{\sc Remark}
\makeatother
\makeatletter
\def\namedlabel#1#2{\begingroup
   \def\@currentlabel{#2}%
   \label{#1}\endgroup}
\makeatother
\date{}
\title{\bf  Expected Confidence Dependency: A Novel Rough Set-Based Approach to Feature Selection}\vspace{.25 in}

\author{
		Saeed Rasouli$^{a,}$\thanks{%
		Corresponding author. E-mail: \href{mailto: srasouli@pgu.ac.ir }{srasouli@pgu.ac.ir }},
	 \ and \
		Hamid  Karamikabir$^{b,}$\thanks{%
		E-mail: \href{mailto: h\_karamikabir@pgu.ac.ir}{h\_karamikabir@pgu.ac.ir}}		\\
		\centerline{$^a$Department of Mathematics, Faculty of Intelligent Systems Engineering and Data Science,}\\
\centerline{Persian Gulf University, Bushehr, Iran}\\
	\centerline{$^b$Department of Statistics,   Faculty of Intelligent Systems Engineering and Data Science,}\\
\centerline{Persian Gulf University, Bushehr, Iran} }

\begin{document}
 \maketitle
\begin{abstract}
This paper proposes Expected Confidence Dependency (ECD), a novel, soft-computing–oriented, accuracy-driven dependency measure for feature selection within the rough set theory framework. Unlike traditional rough-set dependency measures that rely on binary characterizations of conditional blocks, ECD assigns confidence-based contributions to individual equivalence blocks and aggregates them through a normalized expectation operator. We formally establish several desirable properties of ECD, including normalization, compatibility with classical dependency, monotonicity, and invariance under structural and label-preserving transformations. 

To evaluate the practical utility of ECD, we conduct an extensive comparative study against three established rough-set dependency measures—classical, relative, and direct—using four benchmark datasets from the UCI Machine Learning Repository. The experimental results demonstrate that ECD yields a more informative and stable assessment of attribute significance, leading to improved feature selection performance. These findings suggest that ECD provides a theoretically grounded and computationally robust criterion that enhances the reliability of rough-set–based learning and uncertainty modeling.
\end{abstract}


\section{Introduction}\label{sec1}

The exponential growth of data in both volume and dimensionality poses substantial challenges for efficient processing, storage, and analysis. Modern datasets, particularly those arising in bioinformatics, text mining, image recognition, and social network analysis often contain tens or even hundreds of thousands of features~\cite{tan2014introduction, wang2016feature}. Such high-dimensional data are computationally demanding, and the presence of redundant or irrelevant attributes may obscure meaningful patterns, hinder interpretability, and degrade the performance of learning algorithms.

Dimensionality reduction has therefore become indispensable in contemporary data-driven research. Two general approaches are commonly employed: \emph{feature extraction}, which transforms the original feature space into a lower-dimensional latent representation (e.g., principal component analysis or autoencoders), and \emph{feature selection}, which identifies a subset of the most informative features. The latter is particularly attractive because it preserves the semantic meaning of the original attributes while eliminating redundancy and noise~\cite{guyon2003introduction}. A carefully selected subset enhances interpretability and computational efficiency and often improves generalization in tasks such as classification, clustering, and rule induction.

Feature selection proves especially effective in high-dimensional settings~\cite{dessi2015similarity, hong2015using, paul2015simultaneous}, where it can markedly reduce computation time and improve predictive performance. As a result, diverse strategies have been proposed, ranging from filter-based methods that evaluate the relevance of individual features to wrapper and embedded methods that assess subsets in conjunction with learning algorithms~\cite{qian2015mutual, han2015global, wei2015heterogeneous, shi2015semi, moradi2015graph, moradi2015integration, bouhamed2015feature, liu2015class, raza2018parallel, raza2018feature, raza2018heuristic}.

Within this methodological spectrum, \emph{Rough Set Theory} (RST), introduced by Pawlak~\cite{pawlak1982rough, pawlak2007rudiments}, has emerged as a powerful mathematical framework for knowledge representation and uncertainty modeling—two core aspects of soft computing. RST characterizes data through indiscernibility relations, equivalence classes, and approximation spaces, enabling reasoning about vagueness without requiring auxiliary information such as probability distributions or membership functions. A central quantity in RST is the \emph{dependency degree}, which quantifies the extent to which conditional attributes determine a decision attribute. Although classical dependency measures are effective, their binary treatment of positive-region membership limits their sensitivity to uncertainty and noise, both of which are inherent to real-world data.

To overcome these limitations, several extensions have sought to enhance the robustness and computational efficiency of dependency estimation. examples include the \emph{parallel rough set-based dependency calculation}~\cite{raza2018parallel}, \emph{heuristic-based dependency calculation}~\cite{raza2018heuristic}, the \emph{direct dependency calculation}~\cite{raza2018feature}. motivated by these advancements and by the soft computing paradigm that values tolerance to imprecision we revisit the concept of dependency from an expectation-driven perspective.

This paper introduces the \emph{expected confidence dependency} (ECD), a novel generalization of Pawlak’s classical dependency that incorporates expected classification confidence across conditional equivalence classes. Unlike the Pawlak traditional dependency, which relies solely on crisp positive-region membership, ECD employs a probabilistic and expectation-based formulation. By weighting classification confidence according to class size and internal consistency, ECD captures partial and uncertain dependencies, yielding a continuous, interpretable, and computationally tractable measure of attribute influence. Consequently, ECD provides a robust and flexible analytical tool for feature selection, rule induction, and knowledge discovery under uncertainty.

The remainder of this paper is organized as follows. Section~\ref{sec2} reviews the mathematical preliminaries of rough set theory, introducing the key notions of information and decision systems, indiscernibility relations, and set approximations. These concepts establish the theoretical groundwork for the subsequent development of dependency measures, illustrated through running examples. Section~\ref{sec3} introduces the proposed \emph{expected confidence dependency} measure as a novel generalization of the classical dependency concept. This measure extends the classical dependency framework by incorporating the concept of \emph{classification confidence} within each conditional equivalence class and averaging these confidences across the entire universe. This section outlines the intuition behind ECD and demonstrates its computation through illustrative examples. Comparative analyses and tabular results highlight the measure’s advantages over the classical, relative, and direct dependencies in both accuracy and robustness. Section~\ref{sec4} formalizes the theoretical underpinnings of the proposed measure and rigorously demonstrates that \(ECD\) satisfies the fundamental properties expected of a dependency function in the RST framework. In particular, Theorem~\ref{fecdpro} establishes \emph{normalization}, \emph{consistency}, the \emph{relation to classical dependency}, and \emph{monotonicity}. These results collectively confirm that \(ECD\) is theoretically sound, mathematically consistent, and computationally reliable for data-driven dependency assessment. Section~\ref{sec5} examines the applicability of \(ECD\) to data analysis tasks, with a primary focus on feature selection. It shows how \(ECD\) enriches rough set-based methodologies by enabling a more nuanced evaluation of attribute relevance and decision consistency. To assess its effectiveness in feature selection, ECD is compared with three established dependency measures classical, relative, and direct using four benchmark datasets from the UCI Machine Learning Repository. The comparative analyses demonstrate that ECD offers a theoretically robust and computationally reliable criterion for data-driven dependency assessment and feature selection. Finally, Section~\ref{sec6} concludes the paper by summarizing the main findings, highlighting theoretical and practical contributions, and outlining directions for future research.

\section{Preliminaries of Rough Set Theory}\label{sec2}

Rough set theory (RST), originally introduced by Pawlak in the early 1980s~\cite{pawlak1982rough, pawlak1991rough}, provides a mathematically rigorous framework for modeling and analyzing uncertainty and vagueness in information systems. Its foundation lies in approximating target sets through indiscernibility relations induced by object attributes. The core idea is to approximate a target concept—represented as a subset of the universe—by means of its lower and upper approximations, defined with respect to equivalence classes of indiscernible objects.

This section revisits the fundamental definitions, properties, and theoretical results of RST that underpin subsequent developments. To enhance interpretability, illustrative examples are included to connect abstract theoretical notions with their practical interpretations.

\subsection{Information Systems}

Following Pawlak~\cite{pawlak1981information}, the concept of an \emph{information system} constitutes the fundamental structure underpinning RST. It provides a formal representation of data by associating each object with a corresponding set of attribute values. Formally, an information system is defined as
\[
K = (U, A, \{V_{a}\}_{a \in A}, f),
\]
where:
\begin{itemize}
    \item $U$ is a non-empty finite set of objects, referred to as the \emph{universe};
    \item $A$ is a non-empty finite set of attributes (or features);
    \item $\{V_{a}\}_{a \in A}$ denotes a family of non-empty sets, where each $V_a$ represents the domain of possible values for attribute $a$;
    \item $f: U \times A \to V = \bigcup_{a \in A} V_a$ is the \emph{information function}, satisfying $f(x, a) \in V_a$ for every $x \in U$ and $a \in A$.
\end{itemize}

\begin{remark}
\begin{itemize}
  \item For notational convenience, the information system $K = (U, A, \{V_{a}\}_{a \in A}, f)$ is often abbreviated as $K = (U, A)$. The information function $f(x, a)$ may equivalently be denoted as $a(x)$ for all $x \in U$ and $a \in A$.
  \item An information system can be represented as a data table in which rows correspond to objects and columns to attributes. Each entry associated with a pair $(x, a)$, where $x \in U$ and $a \in A$, contains the attribute value $a(x)$.
\end{itemize}
\end{remark}

\begin{example}\label{insyex1}
Consider the information system $H = (U, A)$ derived from data concerning fourteen job applicants in a recruitment process. The system encapsulates essential candidate attributes prior to any hiring decision. Table~\ref{taba1} lists the attributes considered, and Table~\ref{tab:tabi1} presents the corresponding information table.

\begin{table}[H]
\footnotesize
\centering
\caption{Attributes used in the hiring process}
\label{taba1}
\vspace{0.5em}
\begin{tabular}{@{}llll@{}}
\toprule
\textbf{Attribute Name} & \textbf{Symbol} & \textbf{Domain} & \textbf{Type} \\ 
\midrule
Experience Level        & Exp    & \{Junior, Mid, Senior\}            & Categorical \\
Education Level         & Edu    & \{HighSchool, Bachelors, Masters\} & Categorical \\
Coding Test Score       & Test   & \{0, 1, 2\}                        & Numerical   \\
Communication Skill     & Comm   & \{0, 1, 2\}                        & Numerical   \\
Willingness to Relocate & Reloc  & \{Yes, No\}                        & Categorical \\
\bottomrule
\end{tabular}
\end{table}

\begin{table}[H]
\footnotesize
\centering
\caption{Information system representing job applicants}
\label{tab:tabi1}
\begin{tabular}{
>{\columncolor[HTML]{9B9B9B}}c llccl}
\hline
Candidate & \multicolumn{5}{c}{\cellcolor[HTML]{C0C0C0}Attributes} \\ \cline{2-6} 
\textbf{} & \textbf{Exp} & \textbf{Edu} & \textbf{Test} & \textbf{Comm} & \textbf{Reloc} \\ \hline
$x_{1}$   & Junior & Bachelors  & 1 & 2 & Yes \\ 
$x_{2}$   & Senior & Masters    & 2 & 1 & No  \\ 
$x_{3}$   & Mid    & Bachelors  & 1 & 1 & Yes \\ 
$x_{4}$   & Junior & HighSchool & 0 & 1 & Yes \\ 
$x_{5}$   & Senior & Masters    & 2 & 2 & No  \\ 
$x_{6}$   & Mid    & Bachelors  & 1 & 1 & No  \\ 
$x_{7}$   & Junior & HighSchool & 0 & 0 & Yes \\ 
$x_{8}$   & Mid    & Masters    & 2 & 2 & Yes \\ 
$x_{9}$   & Senior & Bachelors  & 1 & 0 & No  \\ 
$x_{10}$  & Mid    & HighSchool & 0 & 2 & Yes \\ 
$x_{11}$  & Junior & Masters    & 1 & 2 & Yes \\ 
$x_{12}$  & Senior & Masters    & 2 & 2 & No  \\ 
$x_{13}$  & Mid    & Bachelors  & 1 & 0 & Yes \\ 
$x_{14}$  & Senior & HighSchool & 1 & 1 & No  \\ \hline
\end{tabular}%
\end{table}

This information system enables systematic comparison among candidates with respect to experience, education, technical performance, communication skills, and willingness to relocate. Although no hiring decision is specified, the dataset serves as a foundation for constructing a decision table or applying classification and rule-induction techniques.
\end{example}

\subsection{Decision Systems}

A \emph{decision system} extends an information system by introducing one or more \emph{decision attributes}. A decision attribute—often denoted as a \emph{class}—represents an outcome variable whose value depends on other attributes, called \emph{conditional attributes}. Formally, a decision system is defined as
\[
K = (U, C \cup D),
\]
where $C$ denotes the set of conditional attributes and $D$ denotes the set of decision attributes.

\begin{example}\label{desyex1}
Table~\ref{tab:tabd1} illustrates a decision system derived from Example~\ref{insyex1}. It extends the information system $H$ (Table~\ref{tab:tabi1}) by incorporating a decision attribute, \textbf{Hire}, which indicates whether a given candidate was ultimately employed.

\begin{table}[H]
\footnotesize
\centering
\caption{Decision system representing job applicants}
\label{tab:tabd1}
\begin{tabular}{>{\columncolor[HTML]{9B9B9B}}c cccccc}
\hline
Candidate & \multicolumn{5}{c}{\cellcolor[HTML]{C0C0C0}Conditional Attributes} & \cellcolor[HTML]{EFEFEF}Decision \\ \cline{2-7}
 & \textbf{Exp} & \textbf{Edu} & \textbf{Test} & \textbf{Comm} & \textbf{Reloc} & \textbf{Hire} \\ \hline
$x_{1}$  & Junior & Bachelors  & 1 & 2 & Yes & \cellcolor[HTML]{EFEFEF}Yes \\ 
$x_{2}$  & Senior & Masters    & 2 & 1 & No  & \cellcolor[HTML]{EFEFEF}Yes \\ 
$x_{3}$  & Mid    & Bachelors  & 1 & 1 & Yes & \cellcolor[HTML]{EFEFEF}Yes \\ 
$x_{4}$  & Junior & HighSchool & 0 & 1 & Yes & \cellcolor[HTML]{EFEFEF}No  \\ 
$x_{5}$  & Senior & Masters    & 2 & 2 & No  & \cellcolor[HTML]{EFEFEF}Yes \\ 
$x_{6}$  & Mid    & Bachelors  & 1 & 1 & No  & \cellcolor[HTML]{EFEFEF}No  \\ 
$x_{7}$  & Junior & HighSchool & 0 & 0 & Yes & \cellcolor[HTML]{EFEFEF}No  \\ 
$x_{8}$  & Mid    & Masters    & 2 & 2 & Yes & \cellcolor[HTML]{EFEFEF}Yes \\ 
$x_{9}$  & Senior & Bachelors  & 1 & 0 & No  & \cellcolor[HTML]{EFEFEF}No  \\ 
$x_{10}$ & Mid    & HighSchool & 0 & 2 & Yes & \cellcolor[HTML]{EFEFEF}No  \\ 
$x_{11}$ & Junior & Masters    & 1 & 2 & Yes & \cellcolor[HTML]{EFEFEF}Yes \\ 
$x_{12}$ & Senior & Masters    & 2 & 2 & No  & \cellcolor[HTML]{EFEFEF}No  \\ 
$x_{13}$ & Mid    & Bachelors  & 1 & 0 & Yes & \cellcolor[HTML]{EFEFEF}No  \\ 
$x_{14}$ & Senior & HighSchool & 1 & 1 & No  & \cellcolor[HTML]{EFEFEF}No  \\ \hline
\end{tabular}%
\end{table}

This decision system explicitly distinguishes between conditional and decision attributes, facilitating the application of classification and rule-extraction techniques within the RST framework.
\end{example}


\subsection{Indiscernibility Relation}

The concept of \emph{indiscernibility} formalizes the notion that certain objects cannot be distinguished based on a given subset of attributes. This idea constitutes the cornerstone of RST, forming the basis for set approximations and knowledge granulation.

Let $K = (U, A)$ denote an information system. For any subset $B \subseteq A$, the \emph{indiscernibility relation} induced by $B$ on $K$ is defined as
\[
\mathrm{Ind}_K(B) = \{(x, y) \in U^2 \mid \forall a \in B,\; a(x) = a(y)\}.
\]
Equivalently, this relation can be expressed as the intersection of relations induced by individual attributes:
\[
\mathrm{Ind}_K(B) = \bigcap_{a \in B} \mathrm{Ind}_K(a).
\]

When the underlying system $K$ is clear from context, the subscript is omitted for brevity. The family of all such relations over $K$ is denoted by
\[
\mathrm{Ind}(K) = \{ \mathrm{Ind}(B) \mid \emptyset \neq B \subseteq A \}.
\]

Given $K = (U, A)$ and $B \subseteq A$, the relation $\mathrm{Ind}(B)$ constitutes an equivalence relation on $U$. The equivalence classes generated by $\mathrm{Ind}(B)$ form a partition of $U$, denoted by $U/\mathrm{Ind}(B)$ or simply $U/B$. For any $x \in U$, the equivalence class containing $x$ is denoted by $B(x)$, with alternative notations $[x]_B$ or $[x]_{\mathrm{Ind}(B)}$.

Since empirical data often fail to distinguish between individual objects, reasoning in RST is performed at the level of equivalence classes, which are regarded as \emph{granules of knowledge}~\cite{pal2012rough, pawlakrough, pawlak2007rudiments, polkowski2001rough}. Two objects $x,y \in U$ are said to be \emph{indiscernible with respect to $B$} if $(x,y) \in \mathrm{Ind}(B)$. The equivalence classes, or blocks of $U/B$, are also referred to as \emph{elementary sets} or \emph{elementary granules} determined by $B$.

\begin{example}\label{indecrelex}
Consider the decision system $H$ presented in Example~\ref{desyex1}. 
The following illustrate sample partitions of the universe induced by various subsets of attributes:
{\footnotesize\begin{longtable}{l p{8cm}}
\toprule
\textbf{Partition} & \textbf{Equivalence Classes} \\ 
\midrule
$U/\{\text{Exp}\}$ & $\{x_1, x_4, x_7, x_{11}\}$, $\{x_3, x_6, x_8, x_{10}, x_{13}\}$, $\{x_2, x_5, x_9, x_{12}, x_{14}\}$ \\
$U/\{\text{Edu}\}$ & $\{x_4, x_7, x_{10}, x_{14}\}$, $\{x_1, x_3, x_6, x_9, x_{13}\}$, $\{x_2, x_5, x_8, x_{11}, x_{12}\}$ \\
$U/\{\text{Test}\}$ & $\{x_4, x_7, x_{10}\}$, $\{x_1, x_3, x_6, x_9, x_{11}, x_{13}, x_{14}\}$, $\{x_2, x_5, x_8, x_{12}\}$ \\
$U/\{\text{Comm}\}$ & $\{x_1, x_5, x_8, x_{10}, x_{11}\}$, $\{x_2, x_3, x_4, x_6, x_{14}\}$, $\{x_7, x_9, x_{13}\}$ \\
$U/\{\text{Reloc}\}$ & $\{x_1, x_3, x_4, x_7, x_8, x_{10}, x_{11}, x_{13}\}$, $\{x_2, x_5, x_6, x_9, x_{12}, x_{14}\}$ \\
$U/\{\text{Comm}, \text{Reloc}\}$ & $\{x_1, x_8, x_{10}, x_{11}\}$, $\{x_2, x_6, x_{14}\}$, $\{x_3, x_4\}$, $\{x_5, x_{12}\}$, $\{x_7\}$, $\{x_9, x_{13}\}$ \\
$U/\{\text{Exp}, \text{Edu}, \text{Comm}\}$ & $\{x_3, x_6, x_{13}\}$, $\{x_1, x_{11}\}$, $\{x_2, x_5, x_{12}\}$, $\{x_4\}$, $\{x_7\}$, $\{x_8\}$, $\{x_9\}$, $\{x_{10}\}$, $\{x_{14}\}$ \\
$U/C$ & $\{x_1\}$, $\{x_2\}$, $\{x_3\}$, $\{x_4\}$, $\{x_5\}$, $\{x_6\}$, $\{x_7\}$, $\{x_8\}$, $\{x_9\}$, $\{x_{10}\}$, $\{x_{11}\}$, $\{x_{12}\}$, $\{x_{13}\}$, $\{x_{14}\}$ \\
$U/D$ & $\{x_1, x_2, x_3, x_5, x_8, x_{11}\}$, $\{x_4, x_6, x_7, x_9, x_{10}, x_{12}, x_{13}, x_{14}\}$ \\
\bottomrule
\end{longtable}}
\end{example}

\subsection{Rough Sets and Approximations}

RST provides a formal framework for addressing uncertainty and vagueness by approximating imprecise concepts through their \emph{lower} and \emph{upper} approximations. This formalism underlies numerous applications, including knowledge discovery, decision analysis, and pattern recognition, particularly in systems characterized by incomplete or inconsistent data.

Elementary sets, introduced via indiscernibility relations, constitute the fundamental \emph{granules of knowledge} in RST. Unions of these granules are termed \emph{definable sets} with respect to a chosen attribute subset $B$. Formally, let $K = (U, A)$ denote an information system and $B \subseteq A$. A subset $X \subseteq U$ is \emph{definable} (or \emph{exact}) with respect to $B$ if it can be represented as a union of elementary sets induced by $B$; otherwise, $X$ is \emph{undefinable} (or \emph{rough}).

\begin{example}
Consider Example~\ref{indecrelex}. The set 
\[
X = \{x_2, x_3, x_4, x_6, x_{14}\}
\]
is \emph{exact} with respect to $\{\text{Comm}, \text{Reloc}\}$, as it can be expressed as a union of equivalence classes generated by these attributes. However, $X$ is \emph{rough} with respect to $\{\text{Test}\}$, since it cannot be expressed in this form.
\end{example}

For any rough set $X$, two definable sets with respect to $B$ are associated with it: the \emph{lower approximation} and the \emph{upper approximation}.  

In this context, the pair $(U, B)$ is termed an \emph{approximation space}, and the rough approximation operator is defined as
\[
\mathrm{Apr}^{U}_{B}(X) = \big( \underline{\mathrm{Apr}}^{U}_{B}(X), \; \overline{\mathrm{Apr}}^{U}_{B}(X) \big),
\]
where
\[
\underline{\mathrm{Apr}}^{U}_{B}(X) = \{ u \in U \mid B(u) \subseteq X \}, 
\qquad
\overline{\mathrm{Apr}}^{U}_{B}(X) = \{ u \in U \mid B(u) \cap X \neq \emptyset \}.
\]
Here, $\underline{\mathrm{Apr}}^{U}_{B}(X)$ denotes the \emph{lower approximation} of $X$, comprising objects that can be certainly classified as belonging to $X$ with respect to $B$, whereas $\overline{\mathrm{Apr}}^{U}_{B}(X)$ represents the \emph{upper approximation}, consisting of all objects that may possibly belong to $X$.  

The set
\[
\mathrm{BND}_{B}(X) = \overline{\mathrm{Apr}}_{B}(X) \setminus \underline{\mathrm{Apr}}_{B}(X)
\]
is termed the \emph{boundary region} of $X$ with respect to $B$. Intuitively, the boundary encompasses objects that cannot be conclusively classified as either belonging to $X$ or to its complement based solely on the attributes in $B$.

\begin{remark}
\begin{itemize}
    \item If $\mathrm{BND}_{B}(X) = \emptyset$, then $X$ is \emph{definable} (or exact) with respect to $B$.  
    \item If $\mathrm{BND}_{B}(X) \neq \emptyset$, then $X$ is \emph{rough} with respect to $B$.  
    \item The cardinality of $\mathrm{BND}_{B}(X)$ quantifies the degree of vagueness: a larger boundary implies greater uncertainty in classifying the elements of $U$ relative to $X$.  
\end{itemize}
\end{remark}

\begin{example}\label{ex:lowerupper}
Consider again the decision system $H$ from Example~\ref{desyex1}.  
Let
\[
X = \{x_{1}, x_{2}, x_{3}, x_{5}, x_{8}, x_{11}\},
\]
representing the set of applicants who were ultimately hired.  

From Example~\ref{indecrelex}, recall that
\[
U/\{\text{Test}\} = 
\big\{
\{x_{4}, x_{7}, x_{10}\},\;
\{x_{1}, x_{3}, x_{6}, x_{9}, x_{11}, x_{13}, x_{14}\},\;
\{x_{2}, x_{5}, x_{8}, x_{12}\}
\big\}.
\]

Thus, the approximations of $X$ with respect to $\{\text{Test}\}$ are given by
\[
\begin{array}{ll}
  \underline{\mathrm{Apr}}_{\{\text{Test}\}}(X) & = \{x_{2}, x_{5}, x_{8}\}, \\[0.5em]
  \overline{\mathrm{Apr}}_{\{\text{Test}\}}(X) & = \{x_{1}, x_{2}, x_{3}, x_{5}, x_{6}, x_{8}, x_{9}, x_{11}, x_{13}, x_{14}\}, \\[0.5em]
  \mathrm{BND}_{\{\text{Test}\}}(X) & = \{x_{1}, x_{3}, x_{6}, x_{9}, x_{11}, x_{13}, x_{14}\}.
\end{array}
\]

In this case, the lower approximation identifies applicants who can be certainly classified as hired based solely on their test scores, since every applicant in their equivalence class was hired.  
The upper approximation captures all applicants who may possibly belong to the hired group.  
The boundary region includes those applicants for whom the test score alone does not provide sufficient evidence for a definitive classification.  

Because the boundary region is nonempty, $X$ is a rough set with respect to the attribute subset $\{\text{Test}\}$.
\end{example}


\subsection{Positive Region}

Let $K = (U, A)$ denote a decision system, where $U$ is the universe of discourse and $A = C \cup D$ represents the set of all attributes, partitioned into conditional attributes $C$ and decision attributes $D$. The decision attributes induce a partition of $U$ into decision classes:
\[
U / D = \{D_1, D_2, \dots, D_k\},
\]
where each block $D_i \subseteq U$ consists of all objects sharing identical decision values.

The \emph{positive region} of $D$ with respect to $C$ is defined as
\[
\mathrm{POS}_C(D) = \bigcup_{i=1}^k \underline{\mathrm{Apr}}_{C}(D_i),
\]
where $\underline{\mathrm{Apr}}_{C}(D_i)$ denotes the lower approximation of the decision class $D_i$ with respect to the conditional attributes $C$. 

Intuitively, $\mathrm{POS}_C(D)$ contains all objects whose classification can be determined with certainty based solely on $C$. Objects outside the positive region fall either in the \emph{boundary region}, where their membership in a decision class is ambiguous, or in the \emph{negative region}, where they are certainly excluded from a given decision class.

\begin{example}\label{posregex}
Consider the decision system $H$ introduced in Example~\ref{desyex1}. Let $C' = \{\text{Test}\}$ denote a subset of conditional attributes. The partition of $U$ induced by $\{\text{Test}\}$ is given by
\[
U/\{\text{Test}\} = 
\big\{
\{x_4, x_7, x_{10}\},\;
\{x_1, x_3, x_6, x_9, x_{11}, x_{13}, x_{14}\},\;
\{x_2, x_5, x_8, x_{12}\}
\big\}.
\]

The correspondence between these conditional blocks and decision classes is as follows:
\begin{itemize}
  \item The block $\{x_2, x_5, x_8, x_{12}\}$ is partially consistent: $x_2, x_5, x_8$ belong to ``Hire = Yes,'' but $x_{12}$ belongs to ``Hire = No.'' Hence, $\{x_2, x_5, x_8\}$ contribute to the lower approximation of ``Hire = Yes.''
  \item The block $\{x_4, x_7, x_{10}\}$ lies entirely within ``Hire = No,'' contributing wholly to its lower approximation.
  \item The block $\{x_1, x_3, x_6, x_9, x_{11}, x_{13}, x_{14}\}$ intersects both decision classes and thus contributes to neither lower approximation.
\end{itemize}

Accordingly, the lower approximations of the decision classes are:
\[
\begin{array}{ll}
\underline{\mathrm{Apr}}_{\{\text{Test}\}}(\text{Hire = Yes}) &= \{x_2, x_5, x_8\}, \\[0.5em]
\underline{\mathrm{Apr}}_{\{\text{Test}\}}(\text{Hire = No})  &= \{x_4, x_7, x_{10}\}.
\end{array}
\]
Therefore, the positive region determined by the attribute \textbf{Test} is:
\[
\mathrm{POS}_{\{\text{Test}\}}(D) = \{x_2, x_4, x_5, x_7, x_8, x_{10}\}.
\]
In this case, six applicants can be classified with certainty as either ``hired'' or ``not hired'' using only their test scores. The remaining eight applicants ($x_1, x_3, x_6, x_9, x_{11}, x_{12}, x_{13}, x_{14}$) cannot be unambiguously classified, since their test values occur in both decision classes.

Performing similar computations for other attribute subsets yields:
\[
\begin{array}{ll}
\mathrm{POS}_{\{\text{Exp}\}}(D) &= \varnothing, \\[4pt]
\mathrm{POS}_{\{\text{Edu}\}}(D) &= \{x_4, x_7, x_{10}, x_{14}\}, \\[4pt]
\mathrm{POS}_{\{\text{Comm}\}}(D) &= \{x_7, x_9, x_{13}\}, \\[4pt]
\mathrm{POS}_{\{\text{Reloc}\}}(D) &= \varnothing, \\[4pt]
\mathrm{POS}_{\{\text{Test},\,\text{Comm}\}}(D) &= \{x_1, x_2, x_4, x_7, x_9, x_{10}, x_{11}, x_{13}\}, \\[4pt]
\mathrm{POS}_{\{\text{Edu},\,\text{Comm}\}}(D)  &= \{x_1, x_2, x_4, x_7, x_9, x_{10}, x_{13}, x_{14}\}, \\[4pt]
\mathrm{POS}_{\{\text{Exp},\,\text{Edu}\}}(D)   &= \{x_1, x_4, x_7, x_8, x_9, x_{10}, x_{11}, x_{14}\}, \\[4pt]
\mathrm{POS}_{\{\text{Test},\,\text{Reloc}\}}(D) &= \{x_4, x_6, x_7, x_8, x_9, x_{10}, x_{14}\}.
\end{array}
\]
Thus, the attribute \textbf{Test} alone allows approximately half of the applicants to be classified unambiguously, whereas combinations such as $\{\text{Test}, \text{Comm}\}$ or $\{\text{Test}, \text{Reloc}\}$ fully resolve all ambiguities, indicating their strong joint discriminative capacity.
\end{example}


\subsection{Dependency}

The concept of \emph{dependency} in RST quantifies the extent to which the decision attributes $D$ are determined by the conditional attributes $C$. This notion forms a foundational basis for \emph{attribute reduction} and \emph{feature selection}, where the objective is to identify the minimal subset of conditional attributes that preserves the decision-making capability of the entire attribute set.

Let $K = (U, C \cup D)$ be a decision system, and let $R \subseteq C$.

\subsubsection*{Classical Dependency}

Following \cite{pawlak1981information}, the \emph{classical} dependency of $D$ on $R$ is defined as
\begin{equation}\label{depnformu}
\textsc{Cla}(R,D) = \frac{\lvert \mathrm{POS}_R(D) \rvert}{\lvert U \rvert},
\end{equation}
where $\mathrm{POS}_R(D)$ denotes the positive region of $D$ with respect to $R$.

\subsubsection*{Relative Dependency}

Following \cite{han2004feature}, the \emph{relative} dependency of $D$ on $R$ is defined as
\begin{equation}\label{relfor}
\textsc{Rel}(R, D) = \frac{\lvert U / R \rvert}{\lvert U / (R \cup D) \rvert},
\end{equation}
where $U/R$ denotes the partition of $U$ induced by $R$.

\subsubsection*{Direct Dependency}

Following \cite{raza2018feature}, the \emph{direct} dependency of $D$ on $R$ is defined as
\begin{equation}\label{dirfor}
\textsc{Dir}(R, D) = \frac{\text{Total unique classes}}{\lvert U \rvert}=\frac{\lvert U / (R \cup D) \rvert}{\lvert U \rvert}.
\end{equation}

\medskip

If $\textsc{Dep}$ denotes any of the above dependency criteria, then the value $\textsc{Dep}(R, D)$, termed the \emph{degree of dependency}, represents the proportion of objects in $U$ that can be certainly classified using the attribute subset $R$.

Specifically:
\begin{itemize}
  \item $\textsc{Dep}(R, D) = 1$ indicates that $D$ \emph{totally depends} on $R$, meaning all objects are classified without ambiguity.
  \item $0 < \textsc{Dep}(R, D) < 1$ indicates \emph{partial dependency}, where only some objects can be certainly classified.
  \item $\textsc{Dep}(R, D) = 0$ indicates that $D$ is \emph{independent} of $R$.
\end{itemize}

\begin{example}\label{depenexaextended}
For the decision system \(H\) (Example~\ref{desyex1}), let the decision attribute be \textbf{Hire}. Using the positive regions from Example~\ref{posregex} and Eq.~\eqref{depnformu}, the dependency degrees are summarized in Table~\ref{tab:depenresults}.

\begin{table}[h!]
\centering
\caption{Dependency degrees of the decision attribute on various subsets of conditional attributes}
\label{tab:depenresults}
\footnotesize
\begin{tabular}{l|ccccccccc}
 & \rotatebox{90}{$\{\text{Exp}\}$} 
 & \rotatebox{90}{$\{\text{Edu}\}$} 
 & \rotatebox{90}{$\{\text{Test}\}$} 
 & \rotatebox{90}{$\{\text{Comm}\}$} 
 & \rotatebox{90}{$\{\text{Reloc}\}$} 
 & \rotatebox{90}{$\{\text{Test},\,\text{Comm}\}$} 
 & \rotatebox{90}{$\{\text{Edu},\,\text{Comm}\}$} 
 & \rotatebox{90}{$\{\text{Exp},\,\text{Edu}\}$} 
 & \rotatebox{90}{$\{\text{Test},\,\text{Reloc}\}$} \\ \hline
 $\textsc{Cla}(-,D)$ & 0.000 & 0.286 & 0.214 & 0.214 & 0.000 & 0.571 & 0.571 & 0.571 & 0.500 \\
$\textsc{Rel}(-,D)$ & 0.500 & 0.600 & 0.500 & 0.600 & 0.500 & 0.800 & 0.800 & 0.818 & 0.714 \\
$\textsc{Dir}(-,D)$   & 0.429 & 0.357 & 0.429 & 0.357 & 0.286 & 0.714 & 0.714 & 0.785 & 0.500
\end{tabular}
\end{table}

The results indicate that individual attributes such as \emph{Edu} or \emph{Test} enable partial classification, whereas combinations such as \(\{\text{Test},\text{Comm}\}\), \(\{\text{Edu},\text{Comm}\}\), and \(\{\text{Exp},\text{Edu}\}\) substantially increase the dependency, classifying the majority of cases unambiguously. In particular, the pair \(\{\text{Test},\text{Reloc}\}\) attains $\textsc{Cla}=0.50$ (positive-region size $7$), indicating a strong—but not total—dependency: these two attributes alone classify 7 out of 14 objects without ambiguity.

Let $\textsc{Dep}$ denote a dependency criterion. In practical terms, the dependency degree provides a quantitative measure of the discriminative power of attributes. In feature selection, the objective is to identify a minimal subset $C' \subseteq C$ such that
\[
\textsc{Dep}(C', D) = \textsc{Dep}(C, D),
\]
where $C'$ is called a \emph{reduct} with respect to $\textsc{Dep}$. Reducts offer compact and non-redundant representations of the decision system while preserving its full decision-making capability. Consequently, they constitute the theoretical foundation for dimensionality reduction in RST-based learning systems.
\end{example}
\subsection{Supervised Feature Selection}\label{sec5}

Supervised feature selection aims to identify feature subsets that best predict decision classes by quantifying each attribute’s contribution to the target variable. Many real datasets contain irrelevant or redundant features whose removal improves computational efficiency, generalization, and interpretability. Feature utility is evaluated in terms of relevance to the decision attribute and redundancy with respect to other features~\cite{nahar2013computational}. In RST, supervised reduction relies on dependency-based evaluation functions. Two widely used strategies are forward selection, which incrementally adds features that maximally increase dependency, and backward elimination, which removes features whose deletion does not reduce it. Forward selection is efficient for small reducts, whereas backward elimination is more suitable when redundancy is high but is computationally heavier. These strategies support the construction of reducts—minimal subsets that preserve the system’s classification power.

Let $K = (U, C \cup D)$ denote a decision system, where $U$ is the universe, $C$ the conditional attributes, and $D$ the decision attributes. Let $\textsc{Dep}(\cdot,D)$ denote a generic dependency measure with respect to $D$.

\subsubsection*{Supervised Forward Selection}

The forward selection algorithm based on $\textsc{Dep}$, denoted $\textsc{FDep}$, computes a reduct without exhaustive enumeration of all subsets, thereby reducing complexity relative to brute-force search. It initializes an empty set $R \leftarrow \emptyset$ and iteratively adds the attribute $a \in C \setminus R$ that maximizes $\textsc{Dep}(R \cup \{a\}, D)$. The process terminates when $\textsc{Dep}(R, D) = \textsc{Dep}(C, D)$.

Conceptually, $\textsc{FDep}$ constructs a minimal subset preserving the dependency of the full conditional set, yielding a reduct (or near-reduct in the presence of local optima). Pseudocode is provided in Algorithm~\ref{alg:fdep}.

\begin{algorithm}[htbp]
\caption{Supervised Forward Selection Based on $\textsc{Dep}$ ($\textsc{FDep}$)}
\label{alg:fdep}
\DontPrintSemicolon
\KwIn{$C$: conditional attribute set; $D$: decision attribute set.}
\KwOut{$R$: selected reduct.}
$R \leftarrow \emptyset$\;
\Repeat{$\textsc{Dep}(R,D) = \textsc{Dep}(C,D)$}{
  $T \leftarrow R$\;
  \ForEach{$a \in (C \setminus R)$}{
    \If{$\textsc{Dep}(R \cup \{a\},D) > \textsc{Dep}(T,D)$}{
      $T \leftarrow R \cup \{a\}$\;
    }
  }
  $R \leftarrow T$\;
}
\Return{$R$}\;
\end{algorithm}

Forward selection is efficient because it evaluates only incremental extensions of the current subset, scales well to high-dimensional data, and often yields compact reducts when many features are irrelevant. Its greedy, deterministic nature ensures reproducibility and strong performance with monotonic dependency measures (e.g., classical or ECD). However, greediness may lead to local optima, particularly when features interact synergistically or when the dependency landscape is non-convex. Since added features cannot be removed, the method may produce non-minimal reducts, and weakly monotonic measures (e.g., relative or direct dependency) can further hinder optimal convergence.

\subsubsection*{Supervised Backward Elimination}

The backward elimination algorithm based on $\textsc{Dep}$, denoted $\textsc{BDep}$, adopts the inverse strategy. It initializes $R \leftarrow C$ and iteratively removes any attribute $a \in R$ for which $\textsc{Dep}(R \setminus \{a\}, D) = \textsc{Dep}(C, D)$. Elimination continues until no further attribute can be removed without reducing dependency. Pseudocode is given in Algorithm~\ref{alg:bdep}.

\begin{algorithm}[htbp]
\caption{Supervised Backward Elimination Based on $\textsc{Dep}$ ($\textsc{BDep}$)}
\label{alg:bdep}
\DontPrintSemicolon
\KwIn{$C$: conditional attribute set; $D$: decision attribute set.}
\KwOut{$R$: selected reduct.}
$R \leftarrow C$\;
\Repeat{no attribute can be removed}{
  \ForEach{$a \in R$}{
    \If{$\textsc{Dep}(R \setminus \{a\}, D) = \textsc{Dep}(C, D)$}{
      $R \leftarrow R \setminus \{a\}$\;
    }
  }
}
\Return{$R$}\;
\end{algorithm}

Backward elimination begins with all attributes and removes those whose deletion does not reduce dependency, naturally capturing feature interactions and mitigating some local-optimality issues of forward selection. It is especially effective when many attributes are redundant or when relevance emerges only after others are removed, and with monotonic measures (e.g., classical dependency or ECD) it often yields minimal reducts. Its main drawback is high computational cost, as each removal requires full dependency reevaluation. With non-monotonic measures, premature removals may discard context-dependent features, leading to suboptimal reducts.



\section{Expected Confidence Dependency}\label{sec3}

The classical, relative, and direct dependency measures in Rough Set Theory each quantify the extent to which conditional attributes support the determination of decision classes, but all share a foundational limitation. The \emph{classical} dependency measure evaluates the proportion of objects that can be \emph{certainly} classified based on the positive region determined by the conditional attributes. The \emph{relative} dependency measure assesses the discriminatory capability of attributes by comparing the granularity of partitions induced by the conditional attributes with that of the combined conditional--decision attributes. The more recent \emph{direct} dependency measure quantifies dependency through the number of distinct equivalence classes induced when the conditional attributes are paired with the decision attribute.  

Despite their differences, all three measures operate in a fundamentally binary manner: an equivalence class is treated as either fully consistent with a single decision class or inconsistent altogether. As a consequence, these criteria cannot adequately capture \emph{partial} or \emph{graded} dependencies—situations frequently encountered in real-world data where noise, overlap, or uncertainty weakens the strict correspondence between conditional and decision attributes.

To address this limitation, we propose a probabilistic generalization termed the \emph{Expected Confidence Dependency} (ECD). This measure extends the existing dependency framework by introducing the notion of \emph{classification confidence} for each conditional equivalence class and aggregating these confidences across the universe. ECD therefore provides a continuous, expectation-based quantification of the dependency between conditional and decision attributes, enabling the representation of both perfectly consistent and partially consistent relationships that the classical, relative, and direct measures cannot fully express.

Let $K = (U, C \cup D)$ be a decision system, and let $R \subseteq C$. For an equivalence class $X \in U / R$ and a decision class $Y \in U / D$, the \emph{classification confidence of $X$ with respect to $Y$} is defined as
\begin{equation}\label{confexp}
\mathrm{Conf}(X \to Y) = \frac{|X \cap Y|}{|X|}.
\end{equation}
Intuitively, $\mathrm{Conf}(X \to Y)$ represents the conditional probability that an object belonging to $X$ is also a member of the decision class $Y$. The following basic properties hold:
\[
0 \leq \mathrm{Conf}(X \to Y) \leq 1, 
\qquad 
\sum_{Y \in U / D} \mathrm{Conf}(X \to Y) = 1.
\]

Because a conditional block $X$ may overlap with multiple decision classes, the maximum achievable certainty within $X$ is captured by
\[
\mathrm{Conf}(X \to D) := \max_{Y \in U / D} \mathrm{Conf}(X \to Y),
\]
which corresponds to the confidence associated with the majority decision in $X$. This quantity reflects the optimal local classification accuracy attainable when only the attributes in $R$ are used.

To aggregate these local confidences over all conditional equivalence classes, we define the \emph{Total Weighted Confidence Mass} of $R$ with respect to $D$ as
\begin{equation}\label{eq:weightedmass}
W(R, D) =
\sum_{X \in U / R}
|X| \cdot
\mathrm{Conf}(X \to D).
\end{equation}
Here, the term $|X| \cdot \mathrm{Conf}(X \to D)$ represents the expected number of objects in block $X$ that can be \emph{optimally classified}—that is, correctly assigned to their majority decision class under a majority-voting rule. Summing these contributions across all blocks yields the total number of objects in $U$ expected to be correctly classified using only the information contained in $R$.

The \emph{Expected Confidence Dependency} of $D$ on $R$ is then defined as
\begin{equation}\label{ecdform}
\textsc{Exp}(R, D) = \frac{W(R, D)}{|U|}.
\end{equation}
Normalization by $|U|$ guarantees that $\textsc{Exp}(R, D) \in [0,1]$, allowing it to be interpreted directly as the \emph{expected classification accuracy} achieved when decisions rely solely on the attributes in $R$.

%
%

Algorithm~\ref{alg:ecd} formalizes the computation of $\textsc{Exp}(C, D)$, outlining a systematic procedure for partitioning the universe, calculating confidence values, and aggregating the overall dependency.

\begin{algorithm}[H]
\footnotesize
\caption{Computation of the Expected Confidence Dependency ($ECD$)}
\label{alg:ecd}
\KwIn{
Decision system $K = (U, C \cup D)$;\\
Conditional attribute set $C$;\\
Decision attribute set $D$.
}
\KwOut{
Expected Confidence Dependency value $\textsc{Exp}(C, D)$.
}

\BlankLine
\textbf{Step 1:} Partition the universe $U$ according to $C$:
\[
U / C = \{ X_1, X_2, \dots, X_s \}.
\]

\textbf{Step 2:} Partition $U$ according to $D$:
\[
U / D = \{ Y_1, Y_2, \dots, Y_t \}.
\]

\textbf{Step 3:} Initialize $W(C, D) \leftarrow 0$.

\textbf{Step 4:}
\ForEach{conditional block $X_i \in U / C$}{
    \textbf{(a)} For each decision block $Y_j \in U / D$, compute:
    \[
    \mathrm{Conf}(X_i \to Y_j) = \frac{|X_i \cap Y_j|}{|X_i|}.
    \]
    \textbf{(b)} Determine the maximum confidence for $X_i$:
    \[
    \mathrm{Conf}(X_i \to D) = \max_{1 \le j \le t} \mathrm{Conf}(X_i \to Y_j).
    \]
    \textbf{(c)} Update:
    \[
    W(C, D) \leftarrow W(C, D) + |X_i| \cdot \mathrm{Conf}(X_i \to D).
    \]
}

\textbf{Step 5:} Compute the ECD:
\[
\textsc{Exp}(C, D) = \frac{W(C, D)}{|U|}.
\]

\textbf{Step 6:} \Return $\textsc{Exp}(C, D)$.
\end{algorithm}

\vspace{5ex}

The computational complexity of Algorithm~\ref{alg:ecd} depends on the numbers of conditional and decision partitions, denoted by $s = |U / C|$ and $t = |U / D|$, respectively. For each conditional block $X_i$, the algorithm computes up to $t$ intersections and confidence values, yielding a total complexity of $\mathcal{O}(s \times t)$. Since $s, t \leq |U|$, the worst-case complexity is bounded by $\mathcal{O}(|U|^2)$. In practice, however, using efficient data structures such as hash tables reduces this to approximately $\mathcal{O}(|U| \cdot |C|)$, particularly when $|C| \ll |U|$. This optimization ensures scalability and makes ECD feasible for large-scale data analysis.

The overall computation procedure is depicted schematically in Fig.~\ref{fig:ecd_flowchart}.

\begin{figure}[H]
\centering
\tikzstyle{startstop} = [rectangle, rounded corners, minimum width=1.5cm, minimum height=0.6cm, text centered, draw=black, fill=gray!20, font=\scriptsize]
\tikzstyle{process} = [rectangle, minimum width=1.8cm, minimum height=0.6cm, text centered, draw=black, fill=blue!10, font=\scriptsize]
\tikzstyle{decision} = [diamond, aspect=1.5, minimum width=1.2cm, minimum height=0.6cm, text centered, draw=black, fill=green!10, font=\scriptsize]
\tikzstyle{arrow} = [thick,->,>=stealth]

\begin{tikzpicture}[node distance=0.5cm and 0.5cm]

\node (start) [startstop] {Start};
\node (input) [process, right=of start] {Input $K$};
\node (partitionC) [process, right=of input] {$U/C = \{X_i\}$};
\node (partitionD) [process, right=of partitionC] {$U/D = \{Y_j\}$};
\node (init) [process, right=of partitionD] {$W(C, D) \leftarrow 0$};

\node (loopC) [process, below=of init] {For $X_i \in U/C$};
\node (loopD) [process, left=of loopC] {For $Y_j \in U/D$};
\node (computeConf) [process, left=of loopD] {Compute $\mathrm{Conf}$};
\node (maxConf) [process, left=of computeConf] {$\max_j \mathrm{Conf}$};
\node (updateTotal) [process, left=of maxConf] {Add $|X_i| \cdot \max$};
\node (decision) [decision, below=of updateTotal] {More $X_i$?};
\node (computeECD) [process, right=of decision, xshift=0.3cm] {Compute $ECD$};
\node (output) [process, right=of computeECD, xshift=0.2cm] {Output $ECD$};
\node (stop) [startstop, right=of output, xshift=0.2cm] {End};

\draw [arrow] (start) -- (input);
\draw [arrow] (input) -- (partitionC);
\draw [arrow] (partitionC) -- (partitionD);
\draw [arrow] (partitionD) -- (init);
\draw [arrow] (init) -- (loopC);
\draw [arrow] (loopC) -- (loopD);
\draw [arrow] (loopD) -- (computeConf);
\draw [arrow] (computeConf) -- (maxConf);
\draw [arrow] (maxConf) -- (updateTotal);
\draw [arrow] (updateTotal) -- (decision);
\draw [arrow] (decision) -- node[above, font=\scriptsize] {No} (computeECD);
\draw [arrow] (decision.south) |- ++(0,-0.2) -| node[below, font=\scriptsize, pos=0.25] {Yes} (loopC.south);
\draw [arrow] (computeECD) -- (output);
\draw [arrow] (output) -- (stop);

\end{tikzpicture}
\caption{Flowchart for computing the ECD.}
\label{fig:ecd_flowchart}
\end{figure}
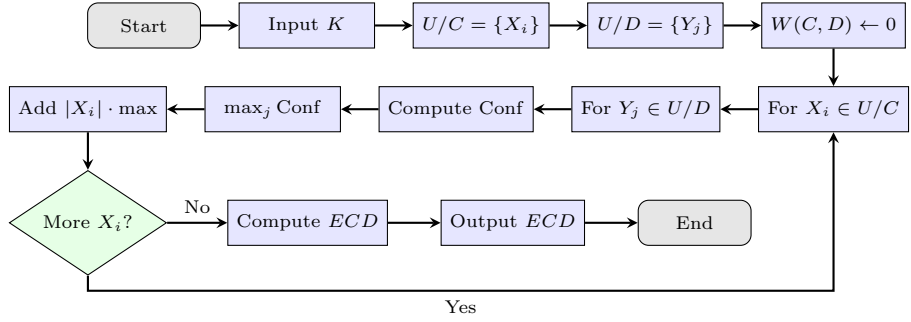

To illustrate the computational procedure and interpretive advantages of the proposed $ECD$ measure, we apply it to the decision system $H$ introduced earlier in Example~\ref{desyex1}. For each conditional attribute, we evaluate the degree of dependency of the decision attribute $D$ on a given subset $R$.

\begin{example}\label{ecdexample}
Consider again the decision system $H$ from Example~\ref{desyex1}. We begin by computing $ECD$ for the single conditional attribute \textit{Test}. Let $R=\{\text{Test}\}$. From the induced partition shown in Example~\ref{indecrelex}, we obtain:
\[
U/R = 
\Big\{
X_{1}=\{x_4, x_7, x_{10}\},\;
X_{2}=\{x_1, x_3, x_6, x_9, x_{11}, x_{13}, x_{14}\},\;
X_{3}=\{x_2, x_5, x_8, x_{12}\}
\Big\}.
\]

\noindent
For each block $X_i$, we identify the majority decision and compute its classification confidence:
\[
\mathrm{Conf}(X_1 \to D)=1,\qquad
\mathrm{Conf}(X_2 \to D)=\frac{4}{7}\approx 0.571,\qquad
\mathrm{Conf}(X_3 \to D)=\frac{3}{4}=0.75.
\]

\noindent
The total weighted confidence mass is therefore
\[
W(R,D) = 3(1) + 7(0.571) + 4(0.75)
        = 3 + 3.997 + 3
        \approx 10.0.
\]
Hence,
\[
\textsc{Exp}(R,D) = \frac{W(R,D)}{|U|} = \frac{10.0}{14} \approx 0.714.
\]

\noindent
Repeating the same procedure for selected attribute subsets (with partitions provided in Example~\ref{depenexaextended}) yields the results summarized in Table~\ref{tab:ecd_vs_gamma}.

\begin{table}[H]
\footnotesize
\centering
\caption{Computation of $ECD$ for selected attribute subsets.}
\label{tab:ecd_vs_gamma}
\renewcommand{\arraystretch}{1.2}
\begin{tabular}{c|ccccccccc}
 & \rotatebox{90}{$\{\text{Exp}\}$} 
 & \rotatebox{90}{$\{\text{Edu}\}$} 
 & \rotatebox{90}{$\{\text{Test}\}$} 
 & \rotatebox{90}{$\{\text{Comm}\}$} 
 & \rotatebox{90}{$\{\text{Reloc}\}$} 
 & \rotatebox{90}{$\{\text{Test},\,\text{Comm}\}$} 
 & \rotatebox{90}{$\{\text{Edu},\,\text{Comm}\}$} 
 & \rotatebox{90}{$\{\text{Exp},\,\text{Edu}\}$} 
 & \rotatebox{90}{$\{\text{Test},\,\text{Reloc}\}$} \\ \hline
$\textsc{Exp}(-,D)$ & 0.571 & 0.785 & 0.714 & 0.714 & 0.571 & 0.857 & 0.857 & 0.857 & 0.857 
\end{tabular}
\end{table}

These results reveal several noteworthy behaviors of the ECD measure. First, ECD consistently exceeds the classical dependency value because it incorporates partial consistency through majority-confidence weighting, whereas the classical measure counts only fully consistent equivalence classes. The attributes \textit{Edu}, \textit{Test}, and \textit{Comm} demonstrate this most clearly, where ECD captures substantial predictive structure that the classical measure overlooks.

Second, for informative combinations such as $\{\text{Test},\text{Comm}\}$, $\{\text{Edu},\text{Comm}\}$, and $\{\text{Exp},\text{Edu}\}$, the ECD values reach their highest nontrivial magnitude (0.857), aligning closely with the relative and direct dependency measures. This indicates that ECD responds appropriately to increased discrimination power as multiple attributes jointly refine the decision partition.

Finally, compared with relative and direct dependency, ECD exhibits smoother and more discriminative variation across attribute subsets. Because it directly reflects the distribution of decision outcomes within each equivalence block, rather than relying solely on partition cardinality ratios, ECD provides finer gradations in partially consistent regions and yields more stable values across similar partitions. These observations demonstrate that ECD encapsulates the underlying decision structure more comprehensively than classical dependency, while offering a robust, expectation-based alternative to the relative and direct measures. Consequently, ECD emerges as a reliable and accuracy-oriented criterion for feature selection in rough set–based data analysis.

Figure~\ref{fig:ecd_heatmap} illustrates the dependency degrees for various attribute subsets in the decision system from Example~\ref{depenexaextended}. As seen in the heatmap, the classical measure (\textsc{Cla}) consistently reports lower dependency values, capturing only fully consistent equivalence classes, whereas the relative (\textsc{Rel}) and direct (\textsc{Dir}) measures yield higher values that reflect partial contributions of attributes. The ECD measure surpasses all three, particularly for informative attribute combinations such as $\{\text{Test},\text{Comm}\}$, $\{\text{Edu},\text{Comm}\}$, and $\{\text{Exp},\text{Edu}\}$, highlighting its ability to integrate majority-confidence weighting and account for partial consistency within each block. The heatmap reveals that ECD provides smoother and more discriminative gradations across attribute subsets, indicating a more nuanced reflection of the decision structure and better capturing the predictive power of individual and combined attributes compared to classical, relative, and direct dependency measures.

\begin{figure}[H]
\centering
\includegraphics[scale=0.45]{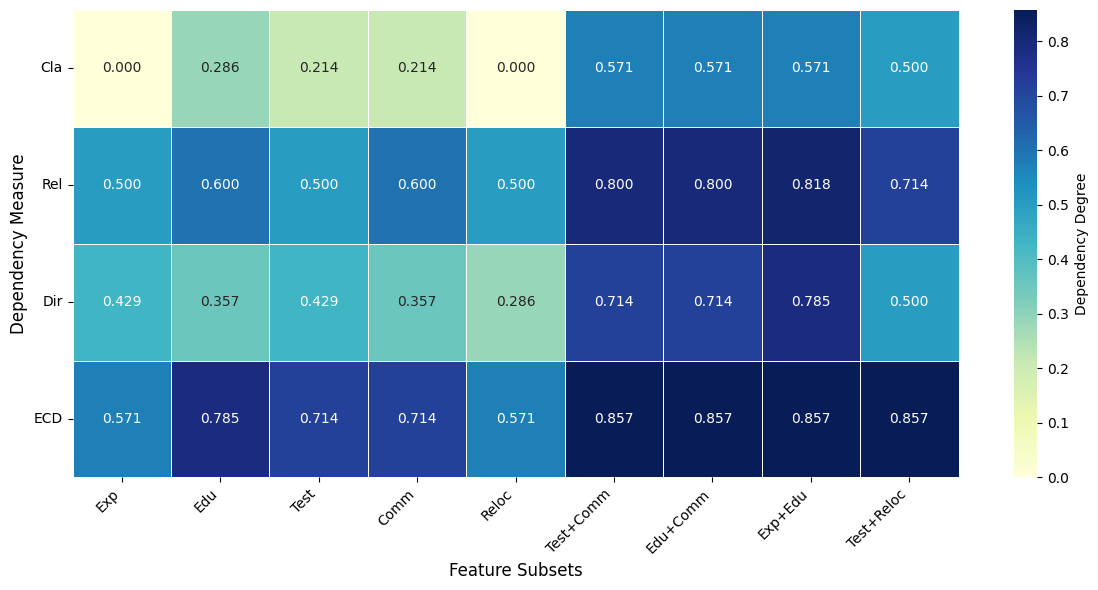} 
\caption{Heatmap of dependency values for selected attribute subsets.}
\label{fig:ecd_heatmap}
\end{figure}

\end{example}


\section{Theoretical Properties of the Expected Confidence Dependency}\label{sec4}

This section formally establishes the principal theoretical foundations of the ECD measure. It rigorously demonstrates that  ECD  satisfies essential mathematical properties expected of a dependency function within RST. These include normalization, consistency, monotonicity, and invariance under structural and labeling transformations. Furthermore, the section provides interpretive, computational, and stability analyses, linking  ECD  to practical classification behavior and robustness. Collectively, these results confirm that  ECD  is both theoretically sound and computationally reliable for data-driven dependency assessment.

\begin{theorem}\label{fecdpro}
  Let $K = (U, C \cup D)$ be a decision system. The following properties hold:
  \begin{enumerate}
  \item [(1) \namedlabel{fecdpro1}{(1)}] $0 \leq \textsc{Exp}(C, D) \leq 1$ for all $C \subseteq A$;
  \item [(2) \namedlabel{fecdpro2}{(2)}] $\textsc{Exp}(C, D) = 1$ if and only if every equivalence class induced by $C$ is fully consistent with a decision class;
  \item [(3) \namedlabel{fecdpro3}{(3)}] $\textsc{Exp}(C, D) \geq \textsc{Cla}(C, D)$, where $\textsc{Cla}(C, D)$ is the classical dependency, with equality if and only if all equivalence classes are consistent;
  \item [(4) \namedlabel{fecdpro4}{(4)}] $\textsc{Exp}(C, D)$ is monotonic with respect to $C$, i.e.,
    \[
    C_1 \subseteq C_2 \implies \textsc{Exp}(C_1, D) \leq \textsc{Exp}(C_2, D).
    \]
\end{enumerate}
\end{theorem}

\begin{proof}
Recall the alternative expression for the Expected Confidence Dependency:
\[
\begin{array}{cc}
  \textsc{Exp}(C,D) & = \frac{1}{|U|}\sum_{[x]_C\in U/C} |[x]_C|\cdot \max_{d\in D}\mathrm{Conf}([x]_C\to d) \\
   & = \frac{1}{|U|}\sum_{[x]_C\in U/C} \max_{d\in D} |[x]_C\cap d|,
\end{array}
\]
since $\mathrm{Conf}([x]_C\to d)=\dfrac{|[x]_C\cap d|}{|[x]_C|}$ and hence $|[x]_C|\cdot\max_{d}\mathrm{Conf}([x]_C\to d)=\max_{d}|[x]_C\cap d|$.

\item [\ref{fecdpro1}:] For every equivalence class $[x]_C$, $0\le |[x]_C\cap d|\le |[x]_C|$ for all $d\in D$, and therefore $0\le \max_{d}|[x]_C\cap d|\le |[x]_C|$. Summing over all classes gives
\[
0 \le \sum_{[x]_C}\max_{d}|[x]_C\cap d| \le \sum_{[x]_C}|[x]_C| = |U|.
\]
Dividing by $|U|$ yields $0\le \textsc{Exp}(C,D)\le 1$.

\item [\ref{fecdpro2}:] Suppose $\textsc{Exp}(C,D)=1$. Then
\[
\sum_{[x]_C}\max_{d}|[x]_C\cap d| = |U|.
\]
Equality holds only if $\max_{d}|[x]_C\cap d|=|[x]_C|$ for all $[x]_C$, meaning each equivalence class is entirely contained within a single decision class. The converse follows analogously.

\item [\ref{fecdpro3}:]  
Let $U/D=\{D_1,\dots,D_k\}$ denote the decision partition and $\mathrm{POS}_C(D)=\bigcup_i\underline{C}(D_i)$ the positive region. By definition,
\[
\textsc{Cla}(C,D)=\frac{|\mathrm{POS}_C(D)|}{|U|}=\frac{1}{|U|}\sum_{\substack{[x]_C\\ [x]_C\subseteq D_i}} |[x]_C|.
\]
For each fully consistent block, $\max_{d}|[x]_C\cap d|=|[x]_C|$, whereas for inconsistent ones, $\max_{d}|[x]_C\cap d|<|[x]_C|$. Consequently,
\[
\sum_{[x]_C}\max_{d}|[x]_C\cap d| \ge \sum_{\substack{[x]_C\\ [x]_C\subseteq D_i}} |[x]_C| = |\mathrm{POS}_C(D)|.
\]
Dividing by $|U|$ gives $\textsc{Exp}(C,D)\ge \textsc{Cla}(C,D)$. Equality holds precisely when all equivalence classes are consistent.

\item [\ref{fecdpro4}:]  
If $C_1\subseteq C_2$, the partition induced by $C_2$ refines that of $C_1$. For each block $B\in U/C_1$, let $B_1,\dots,B_m$ denote its subblocks under $C_2$. Selecting $d^*\in\arg\max_{d}|B\cap d|$, we have
\[
|B\cap d^*| = \sum_{j=1}^m |B_j\cap d^*| \le \sum_{j=1}^m \max_{d}|B_j\cap d|.
\]
Summing over all $B\in U/C_1$ and normalizing by $|U|$ yields $\textsc{Exp}(C_1,D)\le \textsc{Exp}(C_2,D)$, establishing monotonicity.
\end{proof}


We next demonstrate the \emph{partition locality} property of \(ECD\), which asserts that its value depends solely on the partition of the universe induced by the condition attributes, rather than on the specific attributes themselves. This property highlights the structural invariance of \(ECD\) with respect to equivalent attribute representations.

\begin{theorem}\label{prop:partition_locality}
Let $K=(U,C\cup D)$ be a decision system. If two attribute subsets $C_1,C_2\subseteq C$ induce the same partition of $U$, i.e.,
\[
U/C_1=U/C_2,
\]
then
\[
\textsc{Exp}(C_1,D)=\textsc{Exp}(C_2,D).
\]
Thus, $\textsc{Exp}(C,D)$ depends solely on the partition induced by $C$.
\end{theorem}

\begin{proof}
Let the common partition be
\[
\mathcal{B}=U/C_1=U/C_2=\{B_1,\dots,B_s\}.
\]
For any attribute subset $C'$ inducing partition $\mathcal{B}$,
\[
\textsc{Exp}(C',D)=\frac{1}{|U|}\sum_{B\in \mathcal{B}} |B|\cdot \max_{d\in D}\mathrm{Conf}(B\to d)
=\frac{1}{|U|}\sum_{B\in \mathcal{B}} \max_{d\in D}|B\cap d|.
\]
Since this sum depends only on the blocks and their intersections with decision classes, and these are identical for $C_1$ and $C_2$, it follows that $\textsc{Exp}(C_1,D)=\textsc{Exp}(C_2,D)$.
\end{proof}


We now establish the \emph{invariance under relabeling} property of \(ECD\), which ensures that the measure remains invariant under any consistent renaming of attribute values or decision classes. This invariance underscores that \(ECD\) reflects only the underlying informational structure, independent of labeling conventions.

\begin{theorem}\label{prop:invariance_relabel}
Let $K=(U,C\cup D)$ be a decision system.  
\begin{enumerate}
  \item [(1) \namedlabel{prop:invariance_relabel1}{(1)}] Let $\{\pi_a:V_a\to V'_a\}_{a\in C}$ be bijections on the value domains of the conditional attributes, defining a relabeled system $K'=(U,C\cup D)$ with identical decision attributes. Then $U/\mathrm{Ind}_{K'}(C)=U/\mathrm{Ind}_{K}(C)$, and consequently,
  \[
  ECD_{K'}(C,D)=ECD_{K}(C,D).
  \]
  \item [(2) \namedlabel{prop:invariance_relabel2}{(2)}] Let $\sigma:D\to D$ be a permutation of decision-class labels, inducing a relabeled system $K''$. Then
  \[
  ECD_{K''}(C,D)=ECD_{K}(C,D).
  \]
\end{enumerate}
Hence, relabeling attribute values (preserving equality) or permuting decision-class names leaves $ECD$ unchanged.
\end{theorem}

\begin{proof}
\ref{prop:invariance_relabel1}: For any $x,y\in U$ and $a\in C$,
\[
a(x)=a(y)\iff \pi_a(a(x))=\pi_a(a(y)),
\]
implying $\mathrm{Ind}_{K'}(C)=\mathrm{Ind}_{K}(C)$. Thus, $U/C$ and all block–decision intersections remain identical, ensuring $ECD_{K'}(C,D)=ECD_K(C,D)$.

\ref{prop:invariance_relabel2}: Under $\sigma$, each block–decision count satisfies
\[
|C_i \cap \sigma(d)|_{K''} = |C_i \cap d|_{K}.
\]
As the multiset $\{|C_i\cap d|:d\in D\}$ is preserved, $\max_{d}|C_i\cap d|$ remains unchanged. Summing and normalizing yield $ECD_{K''}(C,D)=ECD_K(C,D)$.
\end{proof}


The next proposition demonstrates the \emph{additivity over disjoint universes} property that if the universe is partitioned into non-overlapping subsets, the overall \(ECD\) equals the size-weighted average of the measures computed on each subuniverse.

\begin{theorem}\label{prop:additivity}
Let $K=(U,C\cup D)$ be a decision system where $U=U_1\uplus U_2$ with $U_1\cap U_2=\varnothing$. Denote the induced subsystems as $K_{1}=(U_1,C|_{U_1}\cup D|_{U_1})$ and $K_{2}=(U_2,C|_{U_2}\cup D|_{U_2})$. Then
\[
ECD_{U}(C,D)=\frac{|U_1|}{|U|}\,ECD_{U_1}(C|_{U_1},D|_{U_1})+\frac{|U_2|}{|U|}\,ECD_{U_2}(C|_{U_2},D|_{U_2}).
\]
\end{theorem}

\begin{proof}
Let $n=|U|$, $n_1=|U_1|$, and $n_2=|U_2|$ so that $n=n_1+n_2$. Each equivalence class of $U/C$ lies entirely within either $U_1$ or $U_2$, forming partitions $\mathcal{B}_1$ and $\mathcal{B}_2$. Thus,
\[
ECD_{U}(C,D)=\frac{1}{n}\Bigg(\sum_{B\in\mathcal{B}_1}\max_{d\in D}|B\cap d|+\sum_{B\in\mathcal{B}_2}\max_{d\in D}|B\cap d|\Bigg).
\]
Since each sum corresponds to the unnormalized numerator of $ECD$ on $K_1$ and $K_2$, respectively,
\[
ECD_{U}(C,D)=\frac{n_1}{n}ECD_{U_1}(C|_{U_1},D|_{U_1})+\frac{n_2}{n}ECD_{U_2}(C|_{U_2},D|_{U_2}),
\]
establishing the stated additivity property.
\end{proof}


The next result establishes an interpretive connection between the Expected Confidence Dependency (\(ECD\)) and a concrete classification mechanism. Specifically, it demonstrates that \(ECD\) corresponds to the expected accuracy of a simple majority-vote classifier, thereby providing an intuitive probabilistic interpretation of the measure.

\begin{theorem}\label{prop:majority_accuracy}
Let $K=(U,C\cup D)$ be a decision system and let $U/C=\{C_1,\dots,C_s\}$ denote the partition of $U$ induced by $C$. Define the \emph{block-majority classifier} $h_C:U\to D$ by assigning, for each object $x\in U$, the decision label
\[
h_C(x)\in \arg\max_{d\in D} |C_i\cap d|,
\]
where $C_i$ is the block containing $x$, and ties are broken arbitrarily. Then, the classification accuracy of $h_C$ on $U$, defined as
\[
\mathrm{Acc}(h_C) \;:=\; \frac{1}{|U|}\big|\{\, x\in U \mid h_C(x)=\text{true decision of }x\,\}\big|,
\]
is equal to the Expected Confidence Dependency:
\[
\mathrm{Acc}(h_C) \;=\; \textsc{Exp}(C,D).
\]
Hence, $\textsc{Exp}(C,D)$ represents the expected proportion of correct classifications achieved by the majority-vote classifier that assigns to each conditional block its majority decision label.
\end{theorem}

\begin{proof}
Let $n=|U|$. For each conditional block $C_i\in U/C$ and decision label $d\in D$, denote $n_{i,d}=|C_i\cap d|$ as the number of objects in $C_i$ whose true decision is $d$. Clearly, $|C_i|=\sum_{d\in D} n_{i,d}$.

The block-majority classifier $h_C$ assigns all objects in $C_i$ the label $d_i^\ast\in\arg\max_{d\in D} n_{i,d}$. Therefore, the number of correctly classified objects within $C_i$ is
\[
\max_{d\in D} n_{i,d}.
\]
Summing over all blocks yields the total number of correctly classified objects:
\[
\sum_{i=1}^s \max_{d\in D} n_{i,d}.
\]
Consequently, the classification accuracy of $h_C$ is
\[
\mathrm{Acc}(h_C) \;=\; \frac{1}{n}\sum_{i=1}^s \max_{d\in D} n_{i,d}.
\]

By the integer-count representation of $ECD$,
\[
\textsc{Exp}(C,D) \;=\; \frac{1}{|U|}\sum_{C_i\in U/C} \max_{d\in D} |C_i\cap d|
\;=\; \frac{1}{n}\sum_{i=1}^s \max_{d\in D} n_{i,d}.
\]
Comparing these expressions yields $\mathrm{Acc}(h_C)=\textsc{Exp}(C,D)$.  
The result is unaffected by tie-breaking, since all ties yield the same number $\max_d n_{i,d}$ of correctly classified objects within the tied block, leaving the total accuracy invariant.
\end{proof}


We next provide a computational characterization of \(ECD\). The following proposition shows that \(ECD\) can be expressed entirely in terms of the contingency table describing the joint distribution of conditional and decision classes, making its computation direct and data-driven.

\begin{theorem}\label{prop:contingency_computability}
Let $K=(U,C\cup D)$ be a decision system, and let 
\[
U/C=\{C_1,\dots,C_s\}, \qquad U/D=\{D_1,\dots,D_t\}
\]
denote the partitions of the universe induced by the conditional and decision attributes, respectively. Define the contingency table (count matrix) $N=(n_{ij})\in\mathbb{Z}_{\ge 0}^{s\times t}$ as
\[
n_{ij} := |C_i\cap D_j|, \qquad 1\le i\le s, \; 1\le j\le t.
\]
Then, the Expected Confidence Dependency depends solely on the entries of $N$ and is given by
\[
\textsc{Exp}(C,D) \;=\; \frac{1}{|U|}\sum_{i=1}^s \max_{1\le j\le t} n_{ij}.
\]
Thus, $\textsc{Exp}(C,D)$ is completely determined by the contingency table $N$ and can be computed directly from it without additional information.
\end{theorem}

\begin{proof}
For each conditional block $C_i$ and decision class $D_j$, the confidence is
\[
\mathrm{Conf}(C_i\to D_j)=\frac{|C_i\cap D_j|}{|C_i|}=\frac{n_{ij}}{|C_i|}.
\]
Hence, the maximal confidence for $C_i$ is
\[
\max_{j}\mathrm{Conf}(C_i\to D_j)=\frac{1}{|C_i|}\max_{j} n_{ij}.
\]
Multiplying by $|C_i|$ and summing over all blocks yields
\[
\sum_{i=1}^s |C_i|\cdot \max_{j}\mathrm{Conf}(C_i\to D_j)
= \sum_{i=1}^s \max_{j} n_{ij}.
\]
Normalizing by $|U|$ gives
\[
\textsc{Exp}(C,D)=\frac{1}{|U|}\sum_{i=1}^s \max_{1\le j\le t} n_{ij}.
\]
Since this expression depends only on the entries $n_{ij}$ of the contingency table $N$, the claim follows directly.
\end{proof}


Finally, we analyze the sensitivity of \(ECD\) to perturbations in the decision data. The next result establishes that modifying the decision label of a single object affects the value of \(ECD\) by at most \(1/|U|\), thereby demonstrating the measure’s robustness to minor changes.

\begin{theorem}\label{prop:stability_single_change}
Let $K=(U,C\cup D)$ be a decision system and let $n=|U|$. Suppose $K'$ is obtained from $K$ by changing the decision label of a single object $x\in U$ (while keeping the conditional attributes unchanged). Then
\[
\big|ECD_{K'}(C,D) - ECD_{K}(C,D)\big| \le \frac{1}{n}.
\]
That is, altering the decision label of one object can change $\textsc{Exp}(C,D)$ by no more than $1/|U|$.
\end{theorem}

\begin{proof}
Let $U/C=\{C_1,\dots,C_s\}$ be the partition induced by $C$, and let $x\in C_{i^\ast}$ be the object whose decision label changes from $d_{\mathrm{old}}$ to $d_{\mathrm{new}}$. For each $C_i$ and $d\in D$, define
\[
n_{i,d} = |C_i\cap d| \quad \text{(in $K$)}, \qquad
n'_{i,d} = |C_i\cap d| \quad \text{(in $K'$)}.
\]
Then $n'_{i,d}=n_{i,d}$ for all $i\neq i^\ast$ and $d\in D$, while for the affected block $C_{i^\ast}$,
\[
n'_{i^\ast,d_{\mathrm{old}}} = n_{i^\ast,d_{\mathrm{old}}} - 1, \qquad
n'_{i^\ast,d_{\mathrm{new}}} = n_{i^\ast,d_{\mathrm{new}}} + 1,
\]
and $n'_{i^\ast,d}=n_{i^\ast,d}$ for all other $d$.

Let
\[
M_i = \max_{d\in D} n_{i,d}, \qquad M'_i = \max_{d\in D} n'_{i,d}.
\]
For all $i\neq i^\ast$, we have $M'_i=M_i$. For $i=i^\ast$, since only two entries change by $\pm1$, it follows that
\[
M_{i^\ast}-1 \le M'_{i^\ast} \le M_{i^\ast}+1,
\]
implying $|M'_{i^\ast}-M_{i^\ast}| \le 1$.

The numerator of $ECD$ is the sum of block maxima; thus, the total change in the numerator equals $M'_{i^\ast}-M_{i^\ast}$, whose absolute value is at most $1$. Dividing by $n$ yields
\[
\big|ECD_{K'}(C,D) - ECD_{K}(C,D)\big|
= \frac{1}{n}\big|\sum_i M'_i - \sum_i M_i\big|
= \frac{|M'_{i^\ast}-M_{i^\ast}|}{n}
\le \frac{1}{n}.
\]
This confirms that $ECD$ is stable under a single-label modification.
\end{proof}


\section{Experimental Results}\label{sec5}

In this section, the effectiveness of the proposed dependency criterion for supervised feature selection is evaluated by comparing ECD with three established measures (classical, relative, and direct) using both forward and backward selection across four benchmark datasets from the UCI Machine Learning Repository. These datasets span diverse domains and complexity levels, enabling a comprehensive assessment of robustness and consistency. The best results from each selection method are reported in this paper.
All feature selection procedures are implemented in Python within a unified experimental framework. Categorical attributes are one-hot encoded, and the resulting feature sets are processed by the four algorithms. The selected subsets are subsequently used to train a Random Forest Classifier, and performance is assessed through 5-fold stratified cross-validation.
The following evaluation metrics are computed:
\begin{itemize}
\item \textbf{Accuracy}: proportion of correctly classified instances;
\item \textbf{Precision}: proportion of true positives among predicted positives;
\item \textbf{Recall}: proportion of true positives among actual positives;
\item \textbf{F1-score}: harmonic mean of precision and recall.
\end{itemize}
\subsection{Breast Cancer Dataset}

The Breast Cancer dataset (UCI ID 14) is a classical benchmark dataset for supervised classification tasks. It originates from the University Medical Center in Ljubljana, Slovenia, and is publicly available at: \\ \url{https://archive.ics.uci.edu/dataset/14/breast+cancer}.

The dataset consists of 286 instances and 9 features. The task is binary classification, and the features are primarily categorical, with some ordinal and binary attributes. Features include age, menopause status, tumor size, number of positive lymph nodes, presence of node caps, degree of malignancy, breast side, breast quadrant, and previous irradiation. The target variable indicates whether a patient experienced a recurrence of breast cancer (no-recurrence-events or recurrence-events).

Table~\ref{tab:breastcancer_results} presents the performance of four feature selection algorithms evaluated on the Breast Cancer dataset.

\begin{table}[h!]
\centering
\caption{Performance metrics for feature selection methods on the Breast Cancer dataset.}
\begin{tabular}{lcccc}
\hline
Algorithm & Accuracy & F1-score & Precision & Recall \\
\hline
\textsc{FCla} & 0.706 & 0.691 & 0.684 & 0.706 \\
\textsc{BRel} & 0.703 & 0.682 & 0.675 & 0.703 \\
\textsc{BDir} & 0.703 & 0.682 & 0.675 & 0.703 \\
\textsc{FExp} & 0.731 & 0.711 & 0.708 & 0.731 \\
\hline
\end{tabular}
\label{tab:breastcancer_results}
\end{table}

As observed from Table~\ref{tab:breastcancer_results}, the \textsc{FExp} algorithm achieves the highest performance across all evaluation metrics, with an Accuracy of 0.731, F1-score of 0.711, Precision of 0.708, and Recall of 0.731. This indicates that forward selection based on the proposed dependency criterion (ECD) is more effective in selecting informative features for this dataset compared to classical forward selection (\textsc{FCla}) and backward selection approaches (\textsc{BRel} and \textsc{BDir}). The small differences between \textsc{BRel} and \textsc{BDir} suggest similar behavior of relative and direct dependency measures when used with backward selection. 

The number of features selected varies across algorithms, with \textsc{FCla} selecting 25 features, \textsc{BRel} and \textsc{BDir} selecting 40 features each, and \textsc{FExp} selecting 16 features. Notably, \textsc{FExp} tends to select a smaller subset of informative features, including node-caps, degree of malignancy, tumor size, and breast quadrant, which contributes to its superior classification performance. Overall, the results highlight the robustness of the proposed method in improving classification performance on a benchmark biomedical dataset. The corresponding heatmap illustrating feature dependencies and selection performance is shown in Figure~\ref{breast11}.

\begin{figure}[H]
\centering
\includegraphics[scale=0.7]{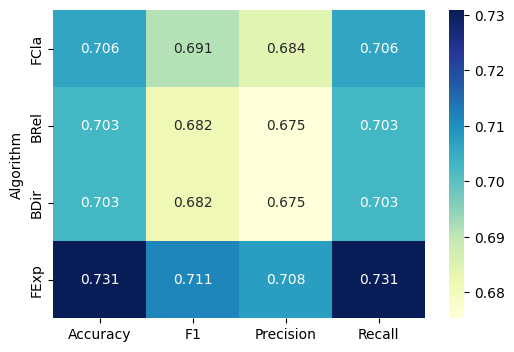}
\caption{Heatmap of performance metrics of feature selection algorithms on the Breast Cancer dataset.}
\label{breast11}
\end{figure}

\subsection{Credit Approval Dataset}

The Credit Approval dataset (UCI ID 27) consists of 690 instances and 15 features, including both categorical and numerical attributes. The binary classification task is to predict whether a credit application is approved. Some features contain missing values and require preprocessing. The dataset is publicly available at: \\ \url{https://archive.ics.uci.edu/dataset/27/credit+approval}.

Table~\ref{tab:credit_results} presents the performance of four feature selection algorithms evaluated on the Credit Approval dataset.

\begin{table}[h!]
\centering
\caption{Performance metrics for feature selection methods on the Credit Approval dataset.}
\begin{tabular}{lcccc}
\hline
Algorithm & Accuracy & F1-score & Precision & Recall \\
\hline
\textsc{FCla} & 0.662 & 0.662 & 0.663 & 0.662 \\
\textsc{BRel} & 0.809 & 0.809 & 0.811 & 0.809 \\
\textsc{BDir} & 0.809 & 0.809 & 0.811 & 0.809 \\
\textsc{FExp} & 0.814 & 0.815 & 0.817 & 0.814 \\
\hline
\end{tabular}
\label{tab:credit_results}
\end{table}

As observed from Table~\ref{tab:credit_results}, the \textsc{FExp} algorithm achieves the highest performance across all evaluation metrics, with an Accuracy of 0.814, F1-score of 0.815, Precision of 0.817, and Recall of 0.814. This indicates that forward selection based on the proposed dependency criterion (ECD) is more effective in selecting informative features for this dataset compared to classical forward selection (\textsc{FCla}) and backward selection approaches (\textsc{BRel} and \textsc{BDir}). The similar performance of \textsc{BRel} and \textsc{BDir} suggests that relative and direct dependency measures behave comparably when used with backward selection.

The number of features selected varies across algorithms, with \textsc{FCla} selecting 3 features, \textsc{BRel} selecting 3 features, \textsc{BDir} selecting 44 features, and \textsc{FExp} selecting 3 features. Notably, \textsc{FExp} tends to select a compact subset of informative features, including A2, A3, and A9, which likely contributes to its superior classification performance. Overall, the results highlight the robustness of the proposed method in improving classification outcomes on a benchmark financial dataset. The corresponding heatmap illustrating feature dependencies and selection performance is shown in Figure~\ref{credit11}.

\begin{figure}[!ht]
\centering
\includegraphics[scale=0.7]{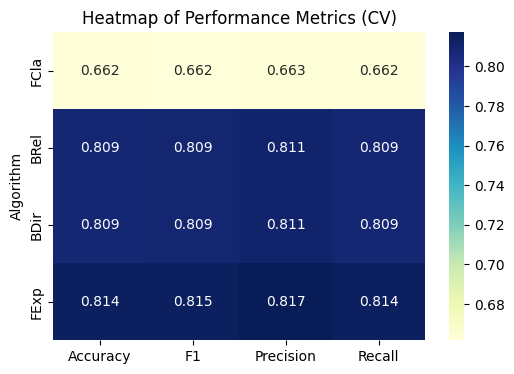}
\caption{Heatmap of performance metrics of feature selection algorithms on the Credit Approval dataset.}
\label{credit11}
\end{figure}

\subsection{Zoo Dataset}

The Zoo dataset (UCI ID 111) contains 101 instances and 17 features, including 16 binary or categorical attributes and one class label for multi-class classification. The features describe animal traits such as hair, feathers, eggs, milk, airborne, aquatic, predator, toothed, backbone, breathes, venomous, fins, legs, tail, domestic, and catsize. The dataset is publicly available at: \\ \url{https://archive.ics.uci.edu/dataset/111/zoo}.

Table~\ref{tab:zoo_results} presents the performance of four feature selection algorithms evaluated on the Zoo dataset.

\begin{table}[h!]
\centering
\caption{Performance metrics for feature selection methods on the Zoo dataset.}
\begin{tabular}{lcccc}
\hline
Algorithm & Accuracy & F1-score & Precision & Recall \\
\hline
\textsc{FCla} & 0.970 & 0.962 & 0.961 & 0.970 \\
\textsc{BRel} & 0.970 & 0.963 & 0.964 & 0.970 \\
\textsc{BDir} & 0.970 & 0.963 & 0.964 & 0.970 \\
\textsc{FExp} & 0.970 & 0.965 & 0.972 & 0.970 \\
\hline
\end{tabular}
\label{tab:zoo_results}
\end{table}

As observed from Table~\ref{tab:zoo_results}, all algorithms achieve very high performance on the Zoo dataset, with Accuracy around 0.97. The \textsc{FExp} algorithm slightly outperforms the others in F1-score and Precision, achieving an F1-score of 0.965 and Precision of 0.972. This demonstrates that forward selection using the proposed dependency criterion (ECD) effectively identifies informative features even in a small multi-class dataset.

The number of features selected varies across algorithms, with \textsc{FCla} selecting 11 features, \textsc{BRel} selecting 6 features, \textsc{BDir} selecting 16 features, and \textsc{FExp} selecting 5 features. Notably, \textsc{FExp} selects a compact set of informative features, including legs, milk, aquatic, toothed, and fins, which contributes to its slightly better performance. Overall, the results highlight the robustness of the proposed method in identifying relevant animal traits and achieving high classification accuracy. The corresponding heatmap illustrating feature dependencies and selection performance is shown in Figure~\ref{zoo_heatmap}.

\begin{figure}[H]
\centering
\includegraphics[scale=0.7]{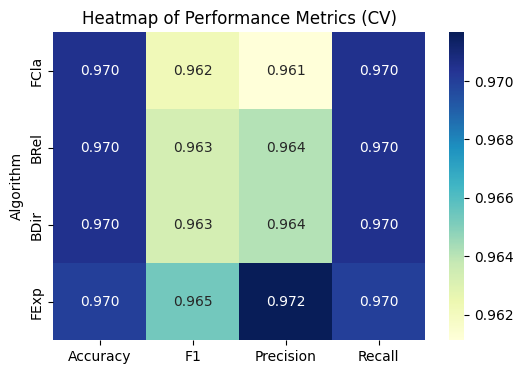}
\caption{Heatmap of performance metrics for feature selection algorithms on the Zoo dataset.}
\label{zoo_heatmap}
\end{figure}

\subsection{Lymphography Dataset}

The Lymphography dataset (UCI ID 63) contains 148 instances and 18 nominal features. The target variable represents four diagnostic classes. The dataset is publicly available at: \\ \url{https://archive.ics.uci.edu/dataset/63/lymphography}.

Table~\ref{tab:lymphography_results} presents the performance of four feature selection algorithms on the Lymphography dataset.

\begin{table}[h!]
\centering
\caption{Performance metrics for feature selection methods on the Lymphography dataset.}
\begin{tabular}{lcccc}
\hline
Algorithm & Accuracy & F1-score & Precision & Recall \\
\hline
\textsc{FCla} & 0.825 & 0.812 & 0.824 & 0.825 \\
\textsc{BRel} & 0.777 & 0.762 & 0.754 & 0.777 \\
\textsc{BDir} & 0.777 & 0.762 & 0.754 & 0.777 \\
\textsc{FExp} & 0.838 & 0.822 & 0.811 & 0.838 \\
\hline
\end{tabular}
\label{tab:lymphography_results}
\end{table}

As observed from Table~\ref{tab:lymphography_results}, the \textsc{FExp} algorithm achieves the highest performance across all evaluation metrics, with an Accuracy of 0.838, F1-score of 0.822, Precision of 0.811, and Recall of 0.838. This indicates that forward selection using the proposed dependency criterion (ECD) effectively identifies informative features for multi-class classification in this dataset, outperforming classical forward selection (\textsc{FCla}) and backward selection approaches (\textsc{BRel} and \textsc{BDir}). 

The number of features selected varies across algorithms, with \textsc{FCla} selecting 11 features, \textsc{BRel} selecting 8 features, \textsc{BDir} selecting 10 features, and \textsc{FExp} selecting 8 features. Notably, \textsc{FExp} selects a compact subset of informative features, including lymphatic-spread, class of tumor, first node status, type of tumor, and patient age (or their original feature names as in the dataset), which contributes to its superior classification performance. Overall, the results demonstrate the robustness of the proposed method in identifying relevant diagnostic features. The corresponding heatmap illustrating feature dependencies and selection performance is shown in Figure~\ref{fig:lymph_heatmap}.

\begin{figure}[H]
\centering
\includegraphics[scale=0.7]{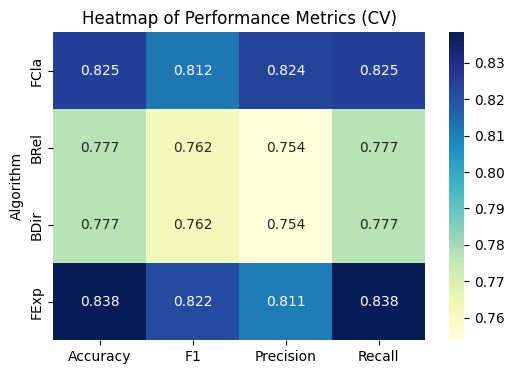}
\caption{Heatmap of performance metrics for feature selection algorithms on the Lymphography dataset.}
\label{fig:lymph_heatmap}
\end{figure}

\section{Conclusion and Future Work}\label{sec6}

This paper introduced the \emph{Expected Confidence Dependency} (ECD), a new dependency measure grounded in rough set theory that incorporates the expected classification confidence of conditional equivalence classes. By moving beyond the binary nature of classical dependency, ECD provides a richer and more nuanced characterization of attribute relevance, particularly in the presence of uncertainty, noise, and partial decision consistency. The formal analysis established that ECD satisfies essential properties—normalization, consistency with Pawlak’s dependency model, monotonicity with respect to attribute expansion, and invariance under structural and labeling transformations—thereby confirming its theoretical soundness and suitability as a robust dependency function.

Empirical experiments conducted on four diverse UCI datasets demonstrated that ECD consistently outperforms classical, relative, and direct dependency measures. The feature subsets selected by ECD improved classification accuracy, maintained stability across cross-validation folds, and exhibited resilience to noisy or weakly informative attributes. These results underscore ECD’s potential as a dependable criterion for supervised feature selection and highlight its usefulness in high-dimensional and heterogeneous data environments.

Several promising directions emerge from this work. Extending ECD to incomplete information systems and incorporating fuzzy or variable-precision mechanisms would enable the model to better accommodate missing values and imprecise boundary regions. Improving scalability—through parallelization, distributed computation, or GPU-based optimization—may also significantly enhance ECD’s applicability to modern large-scale datasets. Another natural direction involves applying ECD to real-world domains such as bioinformatics, cybersecurity, medical diagnostics, and multimodal learning, where uncertainty and heterogeneous feature spaces are prevalent.

Further comparative studies with information-theoretic, statistical, and evolutionary feature selection frameworks may reveal complementary strengths, paving the way for hybrid approaches that combine global search with rough set-based dependency modeling. Extending ECD to multi-label or hierarchical decision systems could broaden its reach to complex decision-making tasks, while deeper investigation into granularity-based interpretations may lead to refined confidence weighting mechanisms. In addition, exploring alternative formulations of equivalence relations—such as similarity-based or distance-aware variants—could yield domain-specific adaptations of ECD with enhanced sensitivity.

Finally, developing open-source software tools and integrating ECD into established machine learning libraries, such as \texttt{scikit-learn}, would facilitate wider adoption by researchers and practitioners. Together, these avenues point toward a rich landscape for future advancements. In conclusion, ECD represents a meaningful step forward in rough set-based feature selection, offering a flexible, interpretable, and theoretically grounded framework with substantial potential for both methodological enhancement and practical deployment.



\end{document}